\documentclass[11pt]{amsart}
\usepackage[foot]{amsaddr}
\usepackage{microtype}
\usepackage{url}
\usepackage{amsthm,amsfonts,amscd,amsmath,amssymb}
\usepackage[foot]{amsaddr}
\usepackage[margin=1in]{geometry}
\usepackage[english]{babel}
\usepackage{mathtools}
\usepackage{graphicx}
\usepackage{bbm}
\usepackage{libertine}

\usepackage{pgfplots}	
\usepgfplotslibrary{fillbetween} 
\pgfplotsset{compat=newest}
\pgfdeclarelayer{background}
\pgfdeclarelayer{foreground}
\pgfsetlayers{background,main,foreground}

\newcommand{\poly}{\mathrm{poly}}
\newcommand{\KS}{\mathrm{KS}}
\newcommand{\e}{\mathrm{e}}
\newcommand{\pnd}{\mathcal{P}_{n,d}}
\newcommand{\dz}{\mathrm{d}z}

\newcommand{\J}{\mathbb{J}}

\newcommand{\R}{\mathbb{R}}

\newcommand{\1}{\mathbbm{1}}
\DeclareMathOperator{\tr}{tr}
\newcommand{\innerprod}[2]{\big\langle #1 , #2 \big\rangle}
\newcommand{\outerprod}[2]{{#1}\otimes{#2}}

\newcommand{\expected}{\mathbb{E}}

\newcommand{\dee}{\,\textup{d}}

\newcommand{\ER}{Erd\H{o}s-R\'{e}nyi\,}
\newcommand{\bm}{\boldsymbol}

\theoremstyle{plain}
	
	\newtheorem{corollary}{Corollary}
	
	\newtheorem{lemma}{Lemma}
	\newtheorem{theorem}{Theorem}
	
\theoremstyle{definition}

\newcommand{\calE}{\mathcal{E}}

\newcommand{\lovasz}{Lov\'{a}sz\,}
\newcommand{\pseudo}{\tilde \expected}

\newcommand{\zee}{\mathfrak{z}}
\DeclareMathOperator{\diag}{diag}
\newcommand{\eps}{\epsilon}

\newcommand{\pbool}{p^{\textup{bool}}}
\newcommand{\psing}{p^{\textup{sing}}}
\newcommand{\pcol}{p^{\textup{col}}}
\newcommand{\pcut}{p^{\textup{cut}}}
\newcommand{\gcut}{g_{\textup{cut}}}

{\begin{array}{l@{\hspace{8mm}}l@{\hspace{8mm}}l}}%
{\end{array}}
\newlength{\lplb}
\begin{document}
\title[The Lov\'asz Theta Function for Random Regular Graphs and Community Detection in the Hard Regime]{The Lov\'asz Theta Function for Random Regular Graphs \\ and Community Detection in the Hard Regime}
\author{Jess Banks}
	\thanks{Dept. of Mathematics, University of California-Berkeley, Berkeley CA}
\author{Robert Kleinberg}
	\thanks{Dept. of Computer Science, Cornell University, Ithaca NY}
\author{Cristopher Moore}
	\thanks{Santa Fe Institute, Santa Fe NM}
	
\maketitle

\begin{abstract}
We derive upper and lower bounds on the degree $d$ for which the Lov\'asz $\vartheta$ function, or equivalently sum-of-squares proofs with degree two, can refute the existence of a $k$-coloring in random regular graphs $G_{n,d}$.  We show that this type of refutation fails well above the $k$-colorability transition, and in particular everywhere below the Kesten-Stigum threshold.  This is consistent with the conjecture that refuting $k$-colorability, or distinguishing $G_{n,d}$ from the planted coloring model, is hard in this region.  Our results also apply to the disassortative case of the stochastic block model, adding evidence to the conjecture that there is a regime where community detection is computationally hard even though it is information-theoretically possible.  Using orthogonal polynomials, we also provide explicit upper bounds on $\vartheta(\overline{G})$ for regular graphs of a given girth, which may be of independent interest.
\end{abstract}


\section{Introduction}

Many constraint satisfaction problems have \emph{phase transitions} in the random case: as the ratio between the number of constraints and the number of variables increases, there is a critical value at which the probability that a solution exists, in the limit $n \to \infty$, suddenly drops from one to zero.  Above this transition, most instances are too constrained and hence unsatisfiable.  But how many constraints do we need before it becomes easy to \emph{prove} that a typical instance is unsatisfiable?  When is there likely to be a short refutation, which we can find in polynomial time, proving that no solution exists?

For a closely related problem, suppose that a constraint satisfaction problem is generated randomly, but with a particular solution ``planted'' in it.  Given the instance, can we recover the planted solution, at least approximately?  For that matter, can we tell whether the instance was generated from this planted model, as opposed to an un-planted model with no built-in solution?  We can think of this as a statistical inference problem.   If there is an underlying pattern in a dataset (the planted solution) but also some noise (the probabilistic process by which the instance is generated) the question is how much data (how many constraints) we need before we can find the pattern, or confirm that one exists. 

Here we focus on the $k$-colorability of random graphs, and more generally the community detection problem. Let $G=G(n,p=d/n)$ denote the \ER graph with $n$ vertices and average degree $d$. A simple first moment argument shows that with high probability $G$ is is not $k$-colorable if 
\begin{equation}
\label{eq:kcol-first}
d \ge d_{\mathrm{first}} 
= 2 k \ln k - \ln k \, . 
\end{equation}
(We say that an event $E_n$ on graphs of size $n$ holds with high probability if $\lim_{n \to \infty} \Pr[E_n] = 1$, and with positive probability if $\liminf_{n \to \infty} \Pr[E_n] > 0$.)  
Sophisticated uses of the second moment method~\cite{achlioptas-naor,coja-oghlan-vilenchik} show that this is essentially tight, and that the $k$-colorability transition occurs at 
\[
d_c = d_{\mathrm{first}} - O(1) \, . 
\]
Now consider the planted coloring model, where we choose a coloring $\sigma$ uniformly at random and condition $G$ on the event that $\sigma$ is proper.  
If $d > d_c$, then $G(n,d/n)$ is probably not $k$-colorable, while graphs drawn from the planted model are $k$-colorable by construction.  Thus, above the $k$-colorability transition, we can tell with high probability whether $G$ was drawn from the planted or un-planted model by checking to see if $G$ is $k$-colorable.  However, searching exhaustively for $k$-colorings would take exponential time.  

A similar situation holds for the stochastic block model, a model of graphs with community structure also known as the planted partition problem (see~\cite{moore-eatcs,abbe-survey} for reviews).  For our purposes, we will define it as follows: fix a constant $\tau$, and say a partition $\sigma$ of the vertices into $k$ groups is ``good'' if a fraction $\tau/k$ of the edges connect vertices within groups.  Equivalently, if $G$ has $m$ edges, $\sigma$ is a multiway cut with $(1-\tau/k)m$ edges crossing between groups.  Generalizing the planted coloring model where $\tau=0$, the block model chooses $\sigma$ uniformly, and conditions $G$ on the event that $\sigma$ is good.  The cases $\tau > 1$ and $\tau < 1$, where vertices are more or less likely to be connected to others in the same group, are called \emph{assortative} (or \emph{ferromagnetic}) and \emph{disassortative} (or \emph{antiferromagnetic}) respectively.  

Two natural problems related to the block model are \emph{detection}, i.e., telling with high probability whether $G$ was drawn from the block model or from $G(n,d/n)$, and \emph{reconstruction}, finding a partition which is significantly correlated with the planted partition $\sigma$.  (This is sometimes called \emph{weak} reconstruction to distinguish it from finding $\sigma$ exactly, which becomes possible when $d=\Theta(\log n)$~\cite{BC09,abbe-bandeira-hall,abbe-sandon,HajekWuXuSDP14,HajekWuXuSDP15,agarwal-etal,mns-consistency}.)  Both problems become information-theoretically possible at a point called the condensation transition~\cite{krzakala-etal-gibbs,coja-etal-cavity,coja-new}, and the first and second moment methods~\cite{banks-etal-colt} show that this scales as
\begin{equation}
\label{eq:block-model-info}
d_c \sim \frac{k \log k}{(\tau-1)^2} \, ,
\end{equation}
where $\sim$ hides a multiplicative constant. As in $k$-coloring this is roughly the first-moment bound above which, with high probability, no good partitions exist in $G(n,d/n)$.  However, the obvious algorithms for detection and reconstruction, such as searching exhaustively for good partitions or sampling from an appropriate Gibbs distribution~\cite{abbe-sandon-isit-2016,abbe-sandon-more-groups-arxiv}, require exponential time.  

In fact, conjectures from statistical physics~\cite{KrzakalaZdeborova09,Decelle2011,Decelle2011a} suggest this exponential difficulty is sometimes unavoidable.  Specifically, these conjectures state that polynomial-time algorithms for detection and reconstruction exist if and only if $d$ is above the \emph{Kesten-Stigum threshold}~\cite{KestenStigum66,KestenStigum66b}, 
\begin{equation}
\label{eq:block-model-ks}
d_\KS = \left( \frac{k-1}{\tau - 1} \right)^2 \, . 
\end{equation}
Several polynomial-time algorithms are now known to succeed whenever $d > d_\KS$, including variants of belief propagation~\cite{mns-colt,abbe-sandon-more-groups} and spectral algorithms based on non-backtracking walks~\cite{mns-proof,non-backtracking,massoulie2014,bordenave-lelarge-massoulie}.  Moreover, for $k=2$ we know that the information-theoretic and Kesten-Stigum thresholds coincide~\cite{mossel-neeman-sly-impossible}.  Comparing~\eqref{eq:block-model-info} and~\eqref{eq:block-model-ks} we see that for any $\tau \ne 1$ we have $d_c < d_\KS$ for sufficiently large $k$, and in fact this occurs for some $\tau < 1$ when $k=4$ and more generally when $k \ge 5$~\cite{abbe-sandon-isit-2016,abbe-sandon-more-groups-arxiv,banks-etal-colt}.  

Thus in the regime $d_c < d < d_\KS$, detection and reconstruction are information-theoretically possible, but are conjectured to be computationally hard.  In particular, this conjecture implies that there is no way to refute the existence of a coloring, or of a good partition, whenever $d < d_\KS$, even when $d$ is large enough so that a coloring or partition probably does not exist.  Our goal in this paper is to rule out spectral refutations based on the \lovasz theta function, or equivalently sum-of-squares proofs of degree two.

For technical reasons, we focus on random $d$-regular graphs, which we denote $G_{n,d}$.  A series of papers applying the first and second moment methods in this setting~\cite{molloy-reed,achlioptas-moore-reg,kemkes-peres-gimenez-wormald,coja-oghlan-efthymiou-hetterich} have determined the likely chromatic number of $G_{n,d}$ for almost all $d$, showing that the critical $d$ for $k$-colorability is $d_c = d_{\mathrm{first}}-O(1)$ just as for $G(n,d/n)$.  (There are a few values of $d$ and $k$ where $G_{n,d}$ could be $k$-colorable with probability strictly between $0$ and $1$, so this transition might not be completely sharp.)  

We define the $d$-regular block model by choosing a planted partition $\sigma$ uniformly at random and conditioning $G_{n,d}$ on the event that $\sigma$ is good.  Equivalently, we choose $G$ uniformly from all $d$-regular graphs such that a fraction $\tau/k$ of their $m=dn/2$ edges connect vertices within groups.  We claim that our results also apply to the regular block model proposed in~\cite{mossel-neeman-sly-impossible} where $d$-regular graphs are chosen with probability proportional to $\tau^{\textrm{\# within-group edges}} ((k-\tau)/(k-1))^{\textrm{\# between-group edges}}$: in that case, the fraction of within-group edges fluctuates, but is $\tau/k + o(1)$ with high probability.\footnote{These models are not to be confused with a stricter model, where for some constants $q_{rs}$ each vertex in group $r$ has exactly $q_{rs}$ neighbors in group $s$~\cite{brito-etal-equitable,coja-mossel-vilenchik,newman-equitable,barucca}.  Our model only constrains the total number of edges within or between groups.}  We again conjecture that refuting the existence of a coloring or a good partition is exponentially hard below the Kesten-Stigum bound.  Since the branching ratio of a $d$-regular tree is $d-1$, in the regular case this becomes
\[
d < d_\KS = \left( \frac{k-1}{\tau - 1} \right)^2 + 1 \, . 
\]

\textbf{Main results.}  The \lovasz $\vartheta$ function, which we review below, gives a lower bound on the chromatic number which can be computed in polynomial time.  In particular, if $\vartheta(\overline{G}) > k$, this provides a polynomial-time refutation of $G$'s $k$-colorability.  We first prove that this type of refutation exactly corresponds to sum-of-squares proofs of degree two in a natural encoding of $k$-colorability as a system of polynomials; this is intuitive, but it does not seem to have appeared in the literature.  We then show the following bounds on the likely value of $\vartheta(\overline{G})$ when $G$ is a random $d$-regular graph.
\begin{theorem}
\label{thm:main-col}
Let $d$ be constant.  For any constant $\eps > 0$, with high probability
	\[ 
\frac{d}{2\sqrt{d-1}} + 1 - \eps
\le \vartheta(\overline{G_{n,d}}) 
\le \frac{d}{2\sqrt{d-1}} + 2 + \eps \, . 
	\]
As a consequence, the \lovasz $\vartheta$ function cannot refute $k$-colorability with high probability if 
\begin{equation}
	\label{eq:lovasz-works} 
	k > 2 + \frac{d}{2\sqrt{d-1}} \, , 
\end{equation}
and in particular if $d$ is below the Kesten-Stigum threshold.
\end{theorem} 
\noindent Rearranging, no refutation of this kind can exist when
$$
	d < 2 (k-2) \left( (k-2) + \sqrt{ (k-2)^2 - 1 } \right) = (4-o_k(1)) d_\KS \, . 
$$
Our lower bound on $\vartheta(\overline{G_{n,d}})$ follows easily from Friedman's theorem~\cite{friedman} on the spectrum of $G_{n,d}$.  For the upper bound, we first use orthogonal polynomials to derive explicit bounds on $\vartheta(\overline{G})$ for arbitrary regular graphs of a given girth---which may be of independent interest---and then employ a concentration argument for $G_{n,d}$.  

We also relate the \lovasz $\vartheta$ function to the existence of a good partition in the disassortative case of the block model, giving
\begin{theorem}
Fix $\tau < 1$ and say a partition is \emph{good} if a fraction $\tau/k$ of its edges connect endpoints in the same group.  Then sum-of-squares proofs of degree two cannot refute the existence of a good partition in $G_{n,d}$ if 
	\[
		\frac{k - \tau}{1 - \tau} > 2 + \frac{d}{2\sqrt{d-1}} \, . 
	\]
\end{theorem}
\noindent 
Thus degree-two sum-of-squares cannot distinguish the regular stochastic block model from $G_{n,d}$ until $d$ is roughly a factor of $4$ above the Kesten-Stigum threshold.

\textbf{Related work.}
The distribution of $\vartheta(\overline{G})$ for the \ER\ graph $G = G(n,p)$ and the random $d$-regular graph $G = G_{n,d}$ were studied in~\cite{coja2003}.  In particular, that work showed that when $d$ is sufficiently large, with high probability $\vartheta(\overline{G_{n,d}}) > c \sqrt{d}$ for a constant $c > 0$.  Our results tighten this lower bound, making the constant $c$ explicit, and provide a nearly-matching upper bound.


Our results on the power of degree-two sum-of-squares refutations for $k$-colorability contribute to a recent line of work on refutations of random CSPs, which we briefly survey. If we define the density of a CSP as the ratio of constraints to variables---which for coloring equals half the average degree of the graph---then the conjectured hard regime for $k$-coloring corresponds to a range of densities bounded below and above by constants (i.e., depending on $k$ but not $n$).  For CSPs such as $k$-SAT and $k$-XOR, there is again a satisfiability transition at constant density, but with high probability sum-of-squares refutations with constant degree do not exist unless the density is much higher, namely $\Omega(n^{k/2 - 1})$~\cite{schoenebeck}, a result which was recently extended to general CSPs whose constraint predicate supports a $(k-1)$-wise uniform distribution~\cite{kmow}.  Conversely, if a predicate does not support a $t$-wise uniform distribution, then~\cite{allen-odonnell-witmer} shows that there is an efficient sum-of-squares refutation when the density is $\tilde{O}(n^{t/2}-1)$.  For coloring, this gives refutations at roughly constant density; our contribution makes this a nearly-precise constant in the special case of degree-two sum-of-squares on random regular graphs.


The hidden clique problem also has a conjectured hard regime.  It is well known that the random graph $G(n,1/2)$ has no cliques larger than $O(\log n)$~\cite{erdos47} but it is conjectured to be computationally hard to distinguish $G(n,1/2)$ from a graph with a planted clique of size $o(n^{1/2})$. A sequence of progressively stronger sum-of-squares lower bounds for this problem~\cite{deshpande15,hopkins2016integrality,meka2015sum} have culminated in the theorem that with high probability the degree-$d$ sum-of-squares proof system cannot refute the existence of a clique of size $n^{1/2 - c(d / \log n)^{1/2}}$ in $G(n,1/2)$ for some constant $c>0$~\cite{bhkkmp16}.

In contrast to the aforementioned work on refuting random $k$-CSPs and planted cliques, our result pertains to a much more specific pair of problems, namely $k$-coloring and the stochastic block model, and only to degree-two sum-of-squares refutations; but it attains a sharp bound, within an additive constant, on the density at which these refutations become possible.  We conjecture that sum-of-squares refutations of any constant degree do not exist below the Kesten-Stigum threshold, but it seems difficult to extend our current techniques to degree higher than two.



\section{Colorings, Partitions, and the \lovasz $\vartheta$ Function}

\subsection{Background on sum-of-squares}

One type of refutation which has gained a great deal of interest recently is sum-of-squares proofs: see~\cite{barak-steurer-review} for a review.  Suppose we encode our variables and constraints as a system of $m$ polynomial equations on $n$ variables, $f_j(x_1,x_2,\ldots,x_n)=0$ for all $j=1,\ldots,m$.   
One way to prove that no solution $\bm{x} \in \R^n$ exists---in algebraic terms, that this variety is empty---is to find a linear combination of the $f_j$ which is greater than zero for all $\bm{x}$.  Moreover, the \emph{positivstellensatz} of Krivine~\cite{Krivine1964} and Stengle~\cite{Stengle1974} shows that a polynomial is nonnegative over $\R^n$ if and only if it is a sum of squares of polynomials.  Thus we need polynomials $g_1, \ldots, g_m$ and $h_1, \ldots, h_t$ and a constant $\eps > 0$ (which we can always scale to $1$ if we like) such that 
\begin{equation}
	\label{eq:sos-refutation}
	\sum_{j=1}^m g_j(\bm{x}) f_j(\bm{x}) = S + \eps 
	\quad \text{where} \quad
	S = \sum_{\ell=1}^t h_\ell(\bm{x})^2 \, . 
\end{equation}
This proof technique is complete as well as sound.  That is, there is such a set of polynomials $\{g_j\}$ and $\{h_\ell\}$ if and only if no solution exists.  

Even when the $f_j$ are of low degree, the polynomials $g_j$ and $h_\ell$ might be of high degree, making them difficult to find.  However, we can ask when a refutation exists where both sides of~\eqref{eq:sos-refutation} have degree $\delta$ or less.  As we take $\delta=2, 4, 6, \ldots$ we obtain the \emph{SOS hierarchy}.  The case $\delta=2$ is typically equivalent to a familiar semidefinite relaxation of the problem.  More generally, a degree-$\delta$ refutation exists if and only if a certain semidefinite program on $O(n^\delta)$ variables is feasible: thus we can find degree-$\delta$ refutations, or confirm that they do not exist, in time $\poly(n^\delta)$~\cite{Shor1987,Nesterov2000,Parrilo2000,Lasserre2001}.  To see why, note that if we write a polynomial $S(\bm{x})$ as a bilinear form on monomials $x^{(\alpha)} = \prod_i x^{\alpha_i}$ of degree $\delta/2$, 
\[
	S(\bm{x}) = \sum_{\alpha,\alpha'} \mathcal{S}(\alpha,\alpha') \,x^{(\alpha)} x^{(\alpha')} \, , 
\]
then $S(\bm{x})$ is a sum of squares of degree $\delta/2$ polynomials if  and only if the matrix $\mathcal{S}$ is positive semidefinite, or equivalently if $\mathcal{S}$ is the sum of positive symmetric rank-one matrices.  These are outer products of vectors with themselves, so there are vectors $w_1,\ldots,w_t$ such that $\mathcal{S} = \sum_{\ell=1}^t \outerprod{w_\ell}{w_\ell}$ and $S = \sum_\ell h_\ell^2$ where $h_\ell(\bm{x}) = \sum_\alpha w_\ell(\alpha) x^{(\alpha)}$.  Finally, the constraint that $S = \sum_j g_j f_j  - \eps$ for some $\{g_j\}$ and some $\eps > 0$ corresponds to a set of linear inequalities on the entries of $\mathcal{S}$.

The dual object to a degree-$\delta$ refutation is a \emph{pseudoexpectation}.  This is a linear operator $\pseudo$ on polynomials of degree at most $\delta$ with the properties that 
\begin{align}
	\pseudo[1] &= 1 \label{eq:pseudo-one} \\
	\pseudo[f_j] &= 0 \text{ for all $j$} \label{eq:pseudo-con} \\
	\pseudo[p^2] &\ge 0 \text{ for any polynomial $p$ of degree at most $\delta/2$.} \label{eq:pseudo-square}
\end{align}
If we write $\pseudo$ as a bilinear form on monomials $x^{(\alpha)}$, then~\eqref{eq:pseudo-one} and~\eqref{eq:pseudo-con} are linear constraints on its entries, and~\eqref{eq:pseudo-square} states that this matrix is positive semidefinite.  The resulting SDP is dual to the SDP for refutations, so each of these SDPs is feasible precisely when the other is not.  Thus there is a degree-$\delta$ refutation if and only if no degree-$\delta$ pseudoexpectation exists, and vice versa.  

We can think of a pseudoexpectation as a way for an adversary to fool the SOS proof system.  The adversary claims there are are many solutions---even if in reality there are none---and offers to compute the expectation of any low-degree polynomial over the set of solutions.  As long as~\eqref{eq:pseudo-one} and~\eqref{eq:pseudo-con} hold, this appears to be a distribution over valid solutions, and as long as~\eqref{eq:pseudo-square} holds, the SOS prover cannot catch the adversary in an obvious lie like the claim that some quantity of degree $\delta/2$ has negative variance.
  


\subsection{Colorings, partitions, and sum-of-squares}

For a given graph $G$ with adjacency matrix $A$, we can encode the problem of $k$-colorability as the following system of polynomial equations in $kn$ variables $\bm{x} = \{x_{i,c}\}$, where $i \in [n]$ indexes vertices and $c \in [k]$ indexes colors:
\begin{align}
	\text{The $x_{i,c}$ are Boolean:}&	& \pbool_{i,c} &\triangleq x_{i,c}^2 - x_{i,c} = 0 &\forall  i,c 
	\label{eq:bool-constraint} \\
	\text{Each vertex has one color:}&	& \psing_{i} &\triangleq -1 + \sum_c x_{i,c} = 0 &\forall  i 
	\label{eq:indicator-constraint} \\
	\text{The coloring is proper:}&	& \pcol_{ij} &\triangleq \sum_c x_{i,c} \,x_{j,c} = 0 &\forall (i,j)\in E 
	\label{eq:col-constraint}
\end{align}
Then $G$ is $k$-colorable if and only if~\eqref{eq:bool-constraint}--\eqref{eq:col-constraint} has a solution in $\R^{kn}$.  
We can encode the stochastic block model similarly: fix $\tau$, and recall that a partition of $G$ into $k$ groups is \emph{good} if a fraction $\tau/k$ of the edges have endpoints in the same group.  If $G$ has $m$ edges, we can replace constraint \eqref{eq:col-constraint} with
\begin{align}
	\text{Good partition:}&	&  \pcut &\triangleq -\frac{\tau}{k} + \frac{1}{2m}\sum_{i,j} A_{ij} \sum_c x_{i,c} \,x_{j,c} = 0 \, . 
	\label{eq:cut-constraint}
\end{align}

A degree-$\delta$ sum-of-squares refutation of \eqref{eq:bool-constraint}--\eqref{eq:col-constraint} is an equation of the form
\begin{align} 
\label{eq:refutation}
	\sum_{i,c} b_{i,c} \pbool_{i,c} + \sum_i s_i \psing_i + \sum_{(i,j)\in E} g_{ij} \pcol_{ij} = S + \eps
\end{align}
where $b_{i,c}, s_i, g_{ij}$ are polynomials over $\bm{x}$, $S$ is a sum of squares of polynomials, $\epsilon$ is a small positive constant which we will omit when clear, and the degree of each side is at most $\delta$. Such an equation is a proof that no coloring exists.  Replacing $\sum_{i,j} g_{ij} \pcol_{ij}$ with $\gcut \pcut$ gives a refutation of the system formed by~\eqref{eq:bool-constraint}, \eqref{eq:indicator-constraint}, and~\eqref{eq:cut-constraint}, proving that no good partition exists.  We focus on refutations of degree two, which as we will see are related to a classic relaxation of graph coloring.

\subsection{The \lovasz $\vartheta$ function}  

An \emph{orthogonal representation} of a graph $G$ with $n$ vertices is an assignment of a unit vector $u_i \in \R^n$ to each vertex $i$ such that $\innerprod{u_i}{u_j} = 0$ for all $(i,j) \in E$.  The \lovasz function, denoted $\vartheta(\overline{G})$ by convention, is the smallest $\kappa$ for which there is an orthogonal representation $\{u_i\}$ and an additional unit vector $\zee \in \R^n$ such that $\innerprod{u_i}{\zee} = 1/\sqrt{\kappa}$: that is, such that all the $u_i$ lie on a cone\footnote{To see that this definition of $\vartheta$ is equivalent to the more common one that $\innerprod{u_i}{\zee} \le 1/\sqrt{\kappa}$ for every $i$, i.e., where the $u_i$ can be in the interior of this cone, simply rotate each $u_i$ in the subspace perpendicular to its neighbors until $\innerprod{u_i}{\zee}$ is exactly $1/\sqrt{\kappa}$.} of width $\cos^{-1} (1/\sqrt{\kappa})$.


The Gram matrix $P_{ij} = \innerprod{u_i}{u_j}$ of an orthogonal representation is positive semidefinite with $P_{ii}=1$ and $P_{ij}=0$ for $(i,j) \in E$. Adding an auxiliary row and column for the inner products with $\zee$, we can define $\vartheta$ in terms of a semidefinite program, 
\begin{align} 
	\label{eq:lovasz-primal}
	\vartheta(\overline{G}) = \min_P \kappa >0
		\qquad \text{such that} \qquad
		 	\begin{pmatrix} 1 & \bm{1}/\sqrt{\kappa} \\ \bm{1}/\sqrt{\kappa} & P \end{pmatrix} &\succeq 0 \\
   P_{ii} &= 1 \qquad \forall i \nonumber \\
   P_{ij} &= 0 \qquad \forall (i,j)\in E \nonumber
\intertext{where $\bm{1}$ is the $n$-dimensional vector whose entries are all $1$s. The dual of this program can be written as}
	\label{eq:lovasz dual}
	\vartheta(\overline{G}) = \max_D \langle D, \J \rangle 
		\qquad \text{such that} \qquad
			D &\succeq 0 \\
\tr D &= 1 \nonumber \\
 D_{ij} &= 0 \qquad \forall (i,j)\notin E \, , \nonumber
\end{align}
where $\J$ is the matrix of all $1$s and $\innerprod{A}{B} = \tr (A^\dagger B) = \sum_{i,j} A_{ij} B_{ij}$ denotes the matrix inner product.

If $G$ is $k$-colorable then $\vartheta(\overline{G}) \le k$, since we can use the first $k$ basis vectors $e_1, \ldots, e_k$ as an orthogonal representation and take $\zee = (1/\sqrt{k}) \sum_{t=1}^k e_t$.  Thus if $\vartheta(\overline{G}) > k$, the \lovasz function gives a polynomial-time refutation of $k$-colorability.  As stated above, degree-two sum-of-squares proofs typically correspond to well-known semidefinite relaxations, and the next theorem shows that this is indeed the case here.  
\begin{theorem} \label{thm:lovasz-col-equivalence}
	There is a degree-2 SOS refutation of $k$-colorability for a graph $G$ if and only if $\vartheta(\overline{G}) > k$.
\end{theorem}
\noindent
We prove this in the Appendix, where we show that any orthogonal representation of $G$ that lies on an appropriate cone lets us define a pseudoexpectation for the system~\eqref{eq:bool-constraint}--\eqref{eq:col-constraint}.  This will also allow us to modify the SDPs for refutations and pseudoexpectations, and work with simplified but equivalent versions.



\subsection{Good partitions and a relaxed \lovasz function}  

The reader may have noticed that while the coloring constraint~\eqref{eq:col-constraint} fixes the inner product $\sum_{c} x_{i,c} x_{j,c} = \innerprod{x_i}{x_j}$ to zero for each edge $(i,j) \in E$, the  ``good partition'' constraint~\eqref{eq:cut-constraint} only fixes the sum of all these inner products.  This suggests a slight relaxation of the \lovasz $\vartheta$ function, where we weaken the SDP~\eqref{eq:lovasz-primal} by replacing the individual constraints on $P_{ij}$ for all $(i,j) \in E$ with a constraint on their sum.  In other words, we allow a vector coloring where neighboring vectors are orthogonal on average.  We denote the resulting function $\hat{\vartheta}$:
\begin{align}
	\label{eq:cut-sdp-primal}
	\hat{\vartheta}(\overline{G}) = \min_{P} \kappa > 0
	\qquad \text{such that} \qquad
	\begin{pmatrix} 1 & \bm{1}/\sqrt{\kappa} \\ \bm{1}/\sqrt{\kappa} & P \end{pmatrix} &\succeq 0 \\
	 P_{ii} &= 1 \qquad \forall i \nonumber \\
	\innerprod{P}{A} &= 0 \, , \nonumber
	\intertext{The dual SDP tightens~\eqref{eq:lovasz dual} by requiring that the matrix $D$ take the same value on every edge.  Thus $D$ is a multiple of $A$ plus a diagonal matrix,}
	\label{eq:cut-sdp-dual}
	\hat \vartheta(\overline{G}) = \max_{\eta,\bm{b}} \langle D, \J \rangle  
	\qquad \text{such that} \qquad
	D \triangleq \eta A + \diag \bm{b} &\succeq 0 \\
	\tr D = \innerprod{\bm{b}}{\bm{1}} &= 1 \nonumber
\end{align}
Since $\hat{\vartheta}$ is a relaxation of $\vartheta$, we always have 
$\hat{\vartheta}(\overline{G}) \le \vartheta(\overline{G})$.

This modified \lovasz function $\hat{\vartheta}$ is equivalent to degree-two SOS for good partitions in the dissasortative case of the block model, in the following sense.
\begin{theorem}
\label{thm:lovasz-cut-equivalence}
If $\tau < 1$, there exists a degree-two SOS refutation of a partition of $G$ where a fraction $\tau/k$ of the edges are within groups if and only if
	\begin{equation} 
		\hat{\vartheta}(\overline{G}) > \frac{k-\tau}{1 - \tau} \, .
	\end{equation}
\end{theorem}
\noindent Once again we leave the proof to the Appendix. Note that the SDP~\eqref{eq:cut-sdp-primal} for $\hat{\vartheta}$ contains no information about $k$ or $\tau$: this relaxed orthogonal representation has the uncanny capacity to fool degree-two SOS about an entire family of related cuts of different sizes and qualities. 

\subsection{Upper and lower bounds} 

With these theorems in hand, we can set about producing degree-two sum-of-squares refutations and pseudoexpectations for our problems; throughout this section we will refer to these simply as `refutations' and `pseudoexpectations'.  In fact, the same construction will give us asymptotically optimal refutations and pseudoexpectations for both the coloring and partition problems.

To warm-up, we have the following simple construction of a refutation, which we will phrase in terms of the \lovasz theta function and its relaxed version.
\begin{lemma}
Let $G$ be a $d$-regular graph, and let $\lambda_{\min}$ be the smallest eigenvalue of its adjacency matrix $A$. Then
	\begin{equation}
		\vartheta(\overline{G}) \ge \hat{\vartheta}(\overline{G}) \ge 1 + d/|\lambda_{\min}| \, .
	\end{equation}
\end{lemma}
\begin{proof}
We construct a feasible solution $D$ to the dual SDP~\eqref{eq:cut-sdp-dual} by taking
$$
	D \triangleq \frac{1}{n}\left(\1 + \frac{1}{|\lambda_{\min}|}A\right) \, , 
$$
and use the fact that $\innerprod{A}{\J} = dn$.
\end{proof}
By invoking Friedman's theorem~\cite{friedman} that (as $n\to \infty$) the smallest eigenvalue of a random $d$-regular graph is with high probability larger than $-2(1 + \eps)\sqrt{d-1}$ for any $\eps>0$, we obtain:
\begin{corollary}
	When $G = G_{n,d}$, for any $\eps > 0$, with high probability
	\begin{equation}
		\vartheta(\overline{G}) \ge \hat{\vartheta}(\overline{G}) > 1 + \frac{d}{2\sqrt{d-1}} - \eps \, .
	\end{equation}
\end{corollary}
\noindent Putting this together with Theorems \ref{thm:lovasz-col-equivalence} and \ref{thm:lovasz-cut-equivalence} gives
\begin{corollary}
If $G = G_{n,d}$ and $\tau < 1$, with high probability there exists a refutation of a partition with a fraction $\tau/k$ of within-group edges when
	\begin{equation}  \label{eq:refutation-exists}
		\frac{k - \tau}{1 - \tau} < 1 + \frac{d}{2\sqrt{d-1}} \, . 
	\end{equation}
Setting $\tau = 0$, a refutation of $k$-colorability exists with high probability when 
\[
k < 1 + \frac{d}{2\sqrt{d-1}} \, . 
\]
\end{corollary}
\noindent 
Note that for large $k$, the minimum value of $d$ satisfying~\eqref{eq:refutation-exists} is a factor of four above the Kesten-Stigum threshold in both the coloring and partition problems.

Our construction for this lower bound on $\vartheta$ is quite simple, but remarkably we find that for both the coloring and partition problems, it is asymptotically optimal in $d$ and $k$. In particular,
\begin{theorem}
	\label{thm:lovasz-upper-bound}
For any $d$-regular graph $G$ with girth at least $\gamma$, we have
	\begin{equation}
	\label{eq:lovasz-upper-bound}
		\hat{\vartheta}(\overline{G}) \le \vartheta(\overline{G}) < 1 + \frac{d}{2(1 - \eps_\gamma)\sqrt{d-1}} \, .
	\end{equation}
	where $\eps_\gamma$ is a sequence of constants which decrease to zero as $\gamma \to \infty$.
\end{theorem}

Since for any constant $\gamma$ a random regular graph has girth $\gamma$ with positive probability~\cite[Theorem 2.12]{wormald-models}, we rely on the following result showing that $\vartheta(\overline{G_{n,d}})$ is concentrated in an interval of width one.  The proof is essentially the same as that of~\cite{achlioptas-moore-reg} for the chromatic number, and is given in the Appendix.

\begin{lemma}
	\label{lem:vartheta-concentrated}
	Let $\theta \ge 3$.  If $\vartheta(\overline{G_{n,d}}) \le \theta$ with positive probability, then $\vartheta(\overline{G_{n,d}}) \le \theta+1$ with high probability.  
\end{lemma}

\begin{corollary}
If $G = G_{n,d}$, with high probability there does not exist a refutation of a partition with a fraction $\tau/k$ of within-group edges when
	\begin{equation}
		\frac{k-\tau}{1 - \tau} > 2 + \frac{d}{2\sqrt{d-1}} \, .
	\end{equation}
Setting $\tau=0$, with high probability no refutation of $k$-colorability exists when 
\[
k > 2 + \frac{d}{2\sqrt{d-1}} \, .
\]
\end{corollary}
\noindent
Thus for both problems, no degree-two sum-of-squares refutation exists until $d$ is roughly a factor of $4$ above the Kesten-Stigum threshold.

\section{Constructing a Pseudoexpectation with Orthogonal Polynomials}
\label{sec:pseudoexpectation}

We now prove Theorem~\ref{thm:lovasz-upper-bound} by constructing a feasible solution to the primal SDP~\eqref{eq:lovasz-primal}: that is, unit vectors $\{u_i\}$ such that $\innerprod{u_i}{u_j} = 0$ for every edge $(i,j)$, and a unit vector $\zee$ so that $\innerprod{u_i}{\zee} = 1/\sqrt{\kappa}$ for all $i$.  Recall that such a collection exists if and only if $\vartheta(\overline{G}) \le \kappa$.

It is convenient to instead define a set of unit vectors $\{v_i\}$ such that $\innerprod{v_i}{v_j} = -1/(\kappa-1)$ for every edge $(i,j)$.  We claim that such a set exists if and only if $\vartheta(\overline{G}) \le \kappa$.  In one direction, given $\{u_i\}$ and $\zee$ with the above properties, if we define
\[
v_i = \sqrt{\frac{\kappa}{\kappa-1}} \,u_i - \frac{1}{\sqrt{\kappa-1}} \,\zee 
\]
then the $v_i$ are unit vectors with $\innerprod{v_i}{v_j} = -1/(\kappa-1)$ for $(i,j) \in E$.  For instance, if the $u_i$ are $k$ orthogonal basis vectors, then the $v_i$ point to the corners of a $k$-simplex.  In the other direction, given $\{v_i\}$ we can take $\zee$ to be a unit vector perpendicular to all the $v_i$, and define
\[
u_i = \sqrt{\frac{\kappa-1}{\kappa}} \,v_i + \frac{1}{\sqrt{\kappa}} \,\zee \, . 
\]
Then $\innerprod{u_i}{u_j} = 0$ for $(i,j) \in E$, and $\innerprod{u_i}{\zee} = 1/\sqrt{\kappa}$ for all $i$.  This means that we can characterize the \lovasz $\vartheta$ function with a slightly different SDP, which uses the Gram matrix of the $\{v_i\}$:
\begin{align}
	\label{eq:lovasz-shift-sdp}
	\vartheta(\overline{G}) = \min_P \, \kappa > 1 
\qquad \text{such that} \qquad 	
P & \succeq 0 \\
P_{ii} &= 1 & \forall i \nonumber \\
P_{ij} &= -1/(\kappa -1) & \forall(i,j) \in E \nonumber
\end{align}

We will show that for any $d$-regular graph $G$ with girth at least $\gamma$, this SDP has a feasible solution with
$$
	\kappa = 1 + \frac{d}{2(1 - \eps_\gamma)\sqrt{d-1}} \, ,
$$
where $\eps_\gamma$ depends only on $\gamma$ and tends to zero as $\gamma \to \infty$.  Therefore, there is a pseudoexpectation that prevents degree-two SOS from refuting $k$-colorability for any $k \ge \kappa$.  We will construct this pseudoexpectation by taking a linear combination of the ``non-backtracking powers'' of $G$'s adjacency matrix $A$.

Denote by $A^{(t)}$ the matrix whose $i,j$ entry is the number of non-backtracking walks of length $t$ from $i$ to $j$; that is, walks which may freely wander the graph so long as they do not make adjacent pairs of steps $a \to b \to a$ for any vertices $a, b$. There is a simple two-term recursion for these matrices: to count non-backtracking walks of length $t+1$, we first extend each walk of length $t$ by one edge, and then subtract off those that backtracked on the last step. This gives
\begin{align}
	A^{(0)} &= \1 \nonumber \\ 
	A^{(1)} &= A \nonumber \\
	A^{(2)} &= A^2 - d\1 \nonumber \\
	A^{(t)} &= A \cdot A^{(t-1)} - (d-1) A^{(t-2)} 
	\quad t \ge 3 \, . 
\end{align}
Borrowing notation from \cite{alon-etal}, we can write $A^{(t)}$ in closed form as
\begin{equation}
\label{eq:aparent}
A^{(t)} = \sqrt{d(d-1)^{t-1}} \,q_t\!\left(\frac{A}{2\sqrt{d-1}}\right) \qquad t\ge 1
\end{equation}
where $q_t(z)$ is a polynomial of degree $t$.  Specifically, 
\begin{align*}
	q_0(z) &= 1 \\ 
	q_1(z) &= 2 \sqrt{\frac{d-1}{d}} z
\end{align*}
and for $t > 1$ the $q_t$ satisfy the Chebyshev recurrence 
\[
q_{t+1}(z) = 2z q_t(z) - q_{t-1}(z) \, . 
\]
We can write $q_t$ explicitly as
\begin{align}
	q_t(z) &= \sqrt{\frac{d-1}{d}} \,U_t(z) - \frac{1}{\sqrt{d(d-1)}} \,U_{t-2}(z) \qquad t \ge 1
	\label{eq:qdef} 
\end{align}
and $U_t$ is the $t$th Chebyshev polynomial of the second kind (note that $U_{-1}(z)=0$).  

Let $\mu(z)$ denote the Kesten-McKay measure $\mu$ on the interval $[-1,+1]$, which after scaling by $2\sqrt{d-1}$ describes the typical spectral density of a random regular graph~\cite{mckay}:
\begin{equation}
\label{eq:kesten-measure}
	\mu(z) = \frac{2}{\pi} \left( \frac{d(d-1)}{d^2-4(d-1)z^2} \right) \sqrt{1-z^2} \, .
\end{equation}
Then the polynomials $q_t$ are orthonormal with respect to this measure.  That is, if we define the inner product
\[
	\langle f, g \rangle = \int f(z) \,g(z) \dee\mu = \int_{-1}^1 f(z) \,g(z) \,\mu(z) \,\dz \, ,
\]
then 
\begin{equation}
\label{eq:orthonormal}
\left\langle q_\ell(z) , q_m(z) \right\rangle = \begin{cases} 1 & \ell=m \\ 0 & \ell \ne m \, . \end{cases} 
\end{equation}

If the girth of the graph is at least $\gamma$, there is no way for a non-backtracking walk of length $\gamma-2$ or less to return to its starting point or to a neighbor of its starting point, so $\innerprod{\1}{A^{(t)}} = \innerprod{A}{A^{(t)}} = 0$ for $1 < t \le \gamma-2$. We can thus satisfy the diagonal and edge constraints of~\eqref{eq:lovasz-shift-sdp} by considering solutions of the form
	\begin{align}
		P &= \1 - \frac{1}{\kappa - 1}A + \sum_{t=2}^{\gamma-2} a_t A^{(t)} \nonumber \\
		&= \1 - \frac{\sqrt{d}}{\kappa - 1} \,q_1\!\left(\frac{A}{2\sqrt{d-1}}\right) 
		+ \sum_{t=2}^{\gamma-2} a_t\sqrt{d(d-1)^{t-1}} \,q_t\!\left(\frac{A}{2\sqrt{d-1}}\right) \label{eq:p-at} \\
		&\triangleq f\!\left(\frac{A}{2\sqrt{d-1}}\right) \, , \nonumber
	\end{align}
since the first two terms ensure that $P$ has $1$s on its diagonal and $-1/(\kappa - 1)$ on the edges.  If we write
\begin{equation}
\label{eq:fdef}
		f(z) = \sum_{t=0}^{\gamma-2} c_t q_t(z) 
		\quad \text{where} \quad 
		c_0 = 1 
		\quad \text{and} \quad 
		c_1 = -\frac{\sqrt{d}}{\kappa - 1} \, ,
\end{equation}
our job is to optimize the coefficients $c_t$ for $1 < t \le \gamma-2$ so as to minimize $c_1$, and hence $\kappa$, while ensuring that $P \succeq 0$.  

The eigenvalues of the matrix $f(A/(2\sqrt{d-1}))$ are of the form $f(\lambda / (2 \sqrt{d-1}))$ where $\lambda$ ranges over all of $A$'s eigenvalues.  Therefore, $P \succeq 0$ if and only if $f(\lambda / (2 \sqrt{d-1}))$ for all eigenvalues $\lambda$ of $A$.  Friedman's celebrated theorem~\cite{friedman} shows that, with high probability, the eigenvalues of $A$ are contained in the set 
\[
	S = \left(-(1+\eps)2\sqrt{d-1},\,(1+\eps)2\sqrt{d-1}\right) \cup \{d\} 
\]
for any $\eps > 0$.  Thus we require that 
\begin{equation}
\label{eq:f-pos}
f(z) \ge 0 \quad \text{for all} \quad
z \in \big( -(1+\eps), 1+\eps \big) \cup \left\{ \frac{d}{2\sqrt{d-1}} \right\} \, . 
\end{equation}
We will relax this condition slightly by demanding just that $f$ is nonnegative on $[-1,+1]$, although as we will see the resulting optimum is achieved by a function which is nonnegative on all of $\R$.  First we use orthonormality~\eqref{eq:orthonormal} to write the coefficients $c_t$ as inner products, 
\[
	c_t = \langle q_t, f \rangle \, .
\]
Then we optimize the pseudoexpectation as follows,
\begin{align} 
\label{eq:polynomial-opt}
	\min \qquad & \langle q_1,f \rangle \\
	\text{such that} \qquad & \langle q_0,f \rangle = 1 \nonumber \\
					& f(z) \ge 0 \qquad \forall z \in [-1,+1] \, . \nonumber
\end{align}

When the degree $\gamma-2$ of $f$ is even, we can solve this optimization problem explicitly.  Set $m = \gamma/2$, and let $r_1 > \cdots > r_m$ be the roots of $q_m$ in decreasing order; it follows from standard arguments about orthogonal polynomials that these are all in the support of $\mu$, i.e., in the interval $[-1,+1]$.  Consider the following polynomial of degree $2(m-1) = \gamma-2$, 
\begin{equation}
	s(z) = \frac{1}{\zeta} \,\prod_{j=1}^{m-1}(z-r_j)^2 \, ,
\end{equation}
where 
\[
	\zeta = \left\langle q_0, \prod_{j=1}^{m-1}(z-r_j)^2 \right\rangle \,  
\]
is a normalizing factor to ensure that $\innerprod{q_0}{s}=1$.  We claim that $s(z)$ is the optimum of~\eqref{eq:polynomial-opt}.  To prove this, we begin with a general lemma on orthogonal polynomials and quadrature.  The proof is standard (e.g.~\cite{szego}) but we include it in the Appendix for completeness.

\begin{lemma} 
\label{lem: orthogonal-quad}
Let $\{p_t\}$ be a sequence of polynomials of degree $t$ which are orthogonal with respect to a measure $\rho$ supported on a compact interval $I$. Then the roots $r_1,\ldots,r_t$ of $p_t$ form a quadrature rule which is exact for any polynomial $u$ of degree less than $2t$, in that
	$$
		\int_I u(z) \dee\rho = \sum_{i=1}^t \omega_i u(r_i)
	$$
for some positive weights $\{\omega_1,\ldots,\omega_t\}$ independent of $u$.
\end{lemma}

Now let $g(z) = z - r_m$. In view of Lemma~\ref{lem: orthogonal-quad}, for any polynomial $f(z)$ of degree at most $\gamma-2$, the inner product $\langle g, f \rangle$ can be expressed using the roots $r_1, \ldots, r_m$ of $q_m$ as a quadrature, 
$$
	\langle g, f \rangle 
	= \int (z - r_m) f(z) \dee\mu 
	= \sum_{j=1}^m \omega_j (r_j - r_m)f(r_j) 
	= \sum_{j=1}^{m-1} \omega_j(r_j - r_m) f(r_j) \, .
$$
Note that $\omega_j(r_j - r_m) > 0$ for every $1 \le j \le m-1$, since $r_m$ is the left-most root.  If impose the constraints that $f(r_j) \ge 0$ for all $j=1,\ldots,m-1$, then $\langle g,f \rangle \ge 0$.  If we also impose the constraint $\langle f, q_0 \rangle = 1$, then
\begin{align}
	\langle q_1, f \rangle 
	&= \left\langle 2\sqrt{\frac{d-1}{d}} \,z, f \right\rangle \nonumber \\
	&= 2\sqrt{\frac{d-1}{d}} \,\langle g, f \rangle + 2\sqrt{\frac{d-1}{d}} \,r_m \langle q_0, f \rangle \nonumber \\
	&\ge 2\sqrt{\frac{d-1}{d}} \,r_m \, , 
\end{align}
with equality if and only if $f(r_j) = 0$ for all $j=1,\ldots,m-1$.  Since $s(z)$ obeys this equality condition, we have
\[
\langle q_1, s \rangle = 2\sqrt{\frac{d-1}{d}} \,r_m \, ,
\]
and this is the minimum possible value of $c_1 = \langle q_1, s \rangle$ subject to the constraints that $\langle q_0, f \rangle = 1$ and $f(r_j) \ge 0$ for $j=1,\ldots,m-1$.  Moreover, $s(z) \ge 0$ on all of $\R$, so $s(z)$ in fact obeys the stronger constraint~\eqref{eq:f-pos}.

Referring back to~\eqref{eq:fdef} gives
\[
c_1 = -\frac{\sqrt{d}}{\kappa - 1} = 2\sqrt{\frac{d-1}{d}} \,r_m \, ,
\]
and so
$$
	\vartheta \ge \kappa = 1 + \frac{d}{2(-r_m)\sqrt{d-1}} \, .
$$
Finally, we obtain~\eqref{eq:lovasz-upper-bound} by defining $\eps_\gamma = r_m +1$.  Since $r_m \to -1$ as $m$ tends to infinity\footnote{%
  The fact that $r_m \to -1$ as $m \to \infty$ can be deduced, for example, 
  by using the definition of $q_m$ in~\eqref{eq:qdef} to observe that 
  $q_m(-1)$ and $q_m(-\cos(\frac{\pi}{m-1}))$ have opposite signs, and
  then applying the intermediate value theorem.
}, we have $\eps_\gamma \to 0$ as $\gamma \to \infty$, completing the proof.

We end with a brief note on the above construction. Recall that our project for the last several pages has been to set the coefficients of non-backtracking paths of length $t$ in a feasible solution $P$ to the SDP~\eqref{eq:lovasz-shift-sdp},
$$
	P = \sum_{t=0}^\gamma a_t A^{(t)}.
$$
As discussed in the Appendix, this matrix can be translated into a degree-two pseudoexpectation $\pseudo$ for the coloring problem: a linear operator that claims to give the joint distribution of colors at at each pair of vertices $i$ and $j$. The reader will find there that $P_{ij}$ is related to the `pseudocorrelation' between vertices $i$ and $j$, by
$$
	\frac{k}{k-1}\left(P_{ij} - 1/k\right) = \widetilde\Pr[\text{$i$ and $j$ are the same color}] \, .
$$ 
Our expansion of $P$ in terms of non-backtracking paths means that, for most pairs $i,j$, this pseudoexpectation depends only on the shortest path distance $d(i,j)$.  Specifically, whenever $d(i,j)=t \le \gamma-2$ and the shortest path is unique, we have $P_{ij} = a_t$, and if $d(i,j) > \gamma-2$ then $P_{ij}=0$.  
One might think that in the limit of large $\gamma$, the optimal pseudoexpectation would make the natural choice that $a_t = (1-k)^{-t}$: in that case, the pseuodocorrelation would decay just as if these shortest paths were colored uniformly at random, ignoring correlations with the remainder of the graph. However, a quick calculation shows that this choice is in fact not optimal.  In fact, the optimal coefficients we derive above cause the pseudocorrelation to decay more quickly with distance than this na\"{i}ve guess, namely (in the limit of large $d$ and large girth) as $a_t \approx t (2(1-k))^{-t}$.




\subsection*{Acknowledgements} We are grateful to Charles Bordenave, Emmanuel Abbe, Amin Coja-Oghlan, Yash Deshpande, Marc Lelarge, and Alex Russell for helpful conversations.  Part of this work was done while C.M. was visiting \'Ecole Normale Sup\'erieure. Part of this work was done while R.K. was a researcher at Microsoft Research New England.  C.M. is supported by the John Templeton Foundation and the Army Research Office under grant W911NF-12-R-0012.

\bibliographystyle{plain}
\bibliography{coloring-sos}

\begin{appendix}
\section{Proof of Theorems~\ref{thm:lovasz-col-equivalence} and~\ref{thm:lovasz-cut-equivalence}}

We prove Theorems~\ref{thm:lovasz-col-equivalence} and~\ref{thm:lovasz-cut-equivalence} by directly simplifying the SDP that defines feasible degree-two pseudoexpectations. The first step is a broad result on the structure of these objects that applies to any set of constraints which includes the boolean~\eqref{eq:bool-constraint} and single-color \eqref{eq:indicator-constraint} constraints and is suitably symmetric; we then specialize to the coloring and partition problems.

Recall that a degree-two pseudoexpectation for a system of polynomials $f_j(\bm{x})=0$ is a linear operator $\pseudo:\R[\bm x]_{\le 2} \to \R$ which satisfies
\begin{itemize}
	\item $\pseudo[1] = 1$
	\item $\pseudo[f_j q] = 0$ for any polynomials $f_j$ and $q$ such that $\deg f_j q \le 2$
	\item $\pseudo[p^2] \ge 0$ for any polynomial $p$ with $\deg p^2 \le 2$
\end{itemize}
We can identify such objects with PSD $(nk + 1) \times (nk + 1)$ matrices of the form
\begin{equation}
\label{eq:onemorerow}
	\pseudo = \begin{pmatrix} 1 & \bm{\ell}^\dagger \\ \bm{\ell} & \calE \end{pmatrix}
\end{equation}
where $\ell_{i,c} = \pseudo[x_{i,c}]$ and $\calE_{(i,c),(j,c')} = \pseudo[x_{i,c} \,x_{j,c'}]$. It is useful to think of $\calE$ as a block matrix, with a $k \times k$ block $\calE_{ij}$ corresponding to each pair of vertices $i,j$.  Consistency with the boolean and single-color constraints~\eqref{eq:bool-constraint}, \eqref{eq:indicator-constraint} then controls the diagonal elements and row and colum sums of each of these blocks,
\begin{align} 
	\calE_{(i,c),(i,c)} &= \pseudo[x_{i,c}^2] = \pseudo[x_{i,c}] = \ell_{i,c} & & \forall i 
	\label{eq:coloring-pseudo} \\
	\sum_{c'} \calE_{(i,c),(j,c')} &= \sum_{c'} \pseudo[x_{i,c} \,x_{j,c'}] = \pseudo[x_{j,c}] = \ell_{i,c} & & \forall i,j
\label{eq:color-pseudo-2}
\end{align}

Moreover, each of our constraints is fixed under permutations of the colors, and $\pseudo$ inherits this symmetry. That is the matrix carries with it a natural $S_k$ action that simultaneously permutes $\pseudo[x_{i,c}] \to \pseudo[x_{i,\sigma(c)}]$ and $\pseudo[x_{i,c} \,x_{j,c'}] \to \pseudo[x_{i,\sigma(c)} \,x_{i,\sigma(c')}]$. This action preserves the spectrum of $\pseudo$ as a matrix, as well as every hard constraint.  By convexity, we may assume that $\pseudo$ is stabilized under it, by beginning with an arbitrary pseudoexpectation and averaging over its orbit.

This assumption substantially constrains and simplifies $\pseudo$. In particular we are free to (i) assume that $\ell_{i,c} = \pseudo[x_{i,c}] = 1/k$ and (ii) assume that each $k\times k$ block in $\calE$ has only two distinct values: one on the diagonal and the other off the diagonal.  In other words, the pseudoexpectation claims that the marginal distribution of each vertex is uniform, and that joint marginal of any two vertices depends only on the probability that they 
have the same or different colors.  
As a result, for each $i,j$ we can assume that $\calE_{ij}$ is a 
linear combination of the identity matrix $\1_k$ and the matrix $\J_k$ of all 1s, and that the row and column sums of $\calE_{ij}$ are all $1/k$.  In that case for each $i,j$ we can write
\begin{equation}
	\calE_{ij} = \frac{1}{k-1}\left(P_{ij} - \frac{1}{k} \right)\left(\1_k - \frac{\J_k}{k} \right) + \frac{\J_k}{k^2} 
\end{equation}
for some $P_{ij}$, or equivalently that
\begin{equation}
\label{eq:calE}
	\calE = \frac{1}{k-1}(P - \J_n/k) \otimes \left(\1_k - \frac{\J_k}{k} \right) + \frac{\J_{nk}}{k^2} 
\end{equation}
for some $n \times n$ matrix $P$.  Note that 
\[
\tr \calE_{ij} = P_{ij} \, ,
\]
so~\eqref{eq:coloring-pseudo} requires that $P_{ii}=1$ for all $i$.

Since the pseudoexpectation~\eqref{eq:onemorerow} consists of $\calE$ with an additional row and column, we consider the following lemma.  We leave its proof as an exercise for the reader.
\begin{lemma} 
\label{lem: augment-psd}
	For any matrix $X$, vector $\bm{v}$ and scalar $b>0$, 
	$$
		\begin{pmatrix} b & \bm{v}^\dagger \\ \bm{v} & X \end{pmatrix} \succeq 0
	$$
	if and only if $X - (1/b)\outerprod{\bm{v}}{\bm{v}} \succeq 0$ \, .
\end{lemma}
\noindent Since $\bm{\ell}$ is the $nk$-dimensional vector whose entries are all $1/k$, we have $\bm{\ell} \otimes \bm{\ell} = \J_{nk}/k^2$.  Thus~\eqref{eq:calE} and Lemma~\ref{lem: augment-psd} imply that $\pseudo \succeq 0$ if and only if 
$$
(P - \J_n/k) \otimes (\1_k - \J_k/k) \succeq 0 \, .
$$ 
Since $\1_k - \J_k/k$ is a projection operator, this in turn occurs if and only if 
\[
P - \J_n/k \succeq 0 \, . 
\]

To summarize, finding a pseudoexpectation is equivalent to finding a PSD matrix $P \in \R^{n\times n}$ with $P_{ii} = 1$ for all $i$, such that $P$ remains PSD when we subtract the rank-one matrix $\J_n/k$.  However, we have thus far only reasoned about the boolean and single color constraints, and including either the coloring or cut constraint places an additional restriction on $P$. In the case of coloring, we demanded that
\begin{equation}
	\sum_c \calE_{(i,c),(j,c)} = \sum_c \pseudo[x_{i,c} \,x_{j,c}] = 0
\end{equation}
for every edge $(i,j)$.  This implies that $\tr \calE_{ij}=0$, and so $P_{ij}=0$ for each edge.  Collecting these observations, a pseudoexpectation for coloring exists exactly when $k > \vartheta(\overline{G})$, where
\begin{align}
	\vartheta(\overline{G}) \triangleq \min_P \, \kappa>0 
		\qquad \text{such that} \qquad
P - \J_n/\kappa & \succeq 0 \\
P_{ii} &= 1 \qquad \forall i \nonumber \\
P_{ij} &= 0 \qquad \forall(i,j) \in E \, . \nonumber
\end{align}
Finally, note that $\J_n/\kappa = v \otimes v$ where $v = \bm{1}_n / \sqrt{\kappa}$.  Applying Lemma~\ref{lem: augment-psd} again then gives exactly the PSD~\eqref{eq:lovasz-primal} for the Lovasz $\vartheta$ function, thus completing the proof of Theorem~\ref{thm:lovasz-col-equivalence}.

In the case of good partitions, we required that
\begin{equation}
	\sum_{i,j} A_{ij} \sum_c \calE_{(i,c),(j,c)} = \sum_{i,j} A_{ij}\sum_c \pseudo[x_{i,c} \,x_{j,c}] = (\tau/k) d n \, ,
\end{equation}
but this means that 
\[
\sum_{i,j} A_{ij} \tr \calE_{ij} 
= \sum_{i,j} A_{ij} P_{ij} 
= \langle P, A \rangle = (\tau/k) d n \, . 
\]
Following the path above, a degree-two pseudoexpectation exists for community detection when $k > \hat{\vartheta}_\tau(\overline{G})$, where
\begin{align}
	\label{eq:hat-theta-tau}
	\hat{\vartheta}_\tau(\overline{G}) \triangleq \min_{P_\tau} \kappa_\tau 
		\qquad \text{such that} \qquad
			P_\tau - \J_n/\kappa_\tau & \succeq 0 \\
 (P_\tau)_{ii} &= 1 \qquad \forall i \nonumber \\
\langle P_\tau,A \rangle &= (\tau/\kappa_\tau) dn \, . \nonumber
\end{align}
\noindent 
A priori, it seems that we may need to solve a different SDP for each value of $\tau$, but a bit more work shows that this is not the case.  Lemma~\ref{lem: augment-psd} lets us transform the SDP~\eqref{eq:cut-sdp-primal} for $\hat{\vartheta}$ to the following problem,
\begin{align}
	\label{eq:hat-theta}
	\hat{\vartheta}(\overline{G}) \triangleq \min_P \kappa 
		\qquad \text{such that} \qquad
P - \J_n/\kappa & \succeq 0 \\
P_{ii} &= 1 \qquad \forall i \nonumber \\
\langle P, A \rangle &= 0 \, . \nonumber
\end{align}
The following lemma then shows us how to relate optima of~\eqref{eq:hat-theta} to those of~\eqref{eq:hat-theta-tau} for any $\tau$ in the disassortative range $\tau < 1$, thus completing the proof of Theorem~\ref{thm:lovasz-cut-equivalence}.

\begin{lemma}
For any $\tau < 1$, 
\begin{equation}
\label{eq:theta-tau}
\hat{\vartheta}(\overline G) = \frac{\hat{\vartheta}_\tau(\overline G) - \tau}{1 - \tau} \, .
\end{equation}
\end{lemma}

\begin{proof}
We show how to translate back and forth between solutions of~\eqref{eq:hat-theta-tau} and~\eqref{eq:hat-theta}.  
Given a matrix $P$, define
\[
	P_\tau = (1-\tau/\kappa_\tau) P + (\tau/\kappa_\tau) \J_n \, .
\]
It is easy to check that $P_{ii} = 1$ if and only if $(P_\tau)_{ii} = 1$, and $\langle P_\tau, A \rangle = (\tau/\kappa_\tau) dn$ if and only if $\langle P, A \rangle = 0$.  Finally, if we set 
\begin{equation}
\label{eq:kappa}
	\kappa = \frac{\kappa_\tau - \tau}{1 - \tau} \, , 
\end{equation}
then 
\[
P_\tau - \J_n / \kappa_\tau 
= (1-\tau/\kappa_\tau) \left( P - \J_n / \kappa \right) \, ,
\]
so $P_\tau - \J_n / \kappa_\tau \succeq 0$ if and only if $P - \J_n / \kappa \succeq 0$.  Thus~\eqref{eq:hat-theta} is feasible for $\kappa$ if and only if~\eqref{eq:hat-theta-tau} is feasible for $\kappa_\tau$.  Since $\hat{\vartheta}(\overline{G})$ and $\hat{\vartheta}_\tau(\overline{G})$ are the smallest $\kappa$ and $\kappa_\tau$ respectively for which this is the case, \eqref{eq:kappa} implies~\eqref{eq:theta-tau}.
\end{proof}

\section{Proof of Lemma~\ref{lem: orthogonal-quad}}

It is immediate that there is such a quadrature rule for polynomials of degree strictly less than $t$, since the space of linear functionals on such polynomials has dimension $t$ and is thus spanned by the $t$ linearly independent functionals which evaluate at the roots $x_i$. Now let $\deg u < 2t$.  We can divide $u$ by $p_t$ to write $u(z) = a(z) p_t + b(z)$ where $\deg a, \deg b < t$. We have
$$
	\int_I u(z) \dee\rho 
	= \int_I \big( a(z) p_t(z) + b(z) \big) \dee\rho 
	= \langle p_t, a \rangle + \int_I b(z) \dee\rho 
	= 0 + \sum_{i=1}^t \omega_i b(r_i) 
	= \sum_{i=1}^t \omega_i u(r_i),
$$
since $p_t$ is orthogonal to all polynomials of degree less than $t$ and has roots $r_i$.  This verifies exactness of the quadrature rule for polynomials of degree smaller than $2t$.  

To show that the weights $\{ \omega_i \}$ are positive, let $i \in \{1,\ldots,t\}$ and let $v_i(z) = (p_t(z)/(z - r_i))^2$ be the polynomial with double roots at every root of $p_t$ save $r_i$.  Since $v_i$ is everywhere nonnegative and is a polynomial of degree $2t-2 < t$, we have
$$
	0 < \int_I v_i(z) \dee\rho = \sum_{j=1}^t \omega_j v_i(r_j) = \omega_i v_i(r_i) \, ,
$$
but since $v(z)$ is nonnegative, $\omega_i$ must be positive.

\section{Proof of Lemma~\ref{lem:vartheta-concentrated}}
\label{sec:vartheta-concentrated}

The proof closely follows~\cite[Theorem 4]{achlioptas-moore-reg} which shows that the chromatic number of $G_{n,d}$ is concentrated on two adjacent integers, and which is in turn based on the proof in~\cite{luczak} of two-point concentration for $G(n,p)$ with $p = O(n^{-5/6-\eps})$.  Recall the configuration model~\cite{wormald-models}, where we make $d$ ``copies'' of each vertex corresponding to its half-edges, and then choose uniformly from all $(dn-1)!! = (dn)! / (2^{dn/2} (dn/2)!)$ perfect matchings of these copies.  If we denote the set of such matchings by $\pnd$ and condition the corresponding multigraphs on having no self-loops or multiple edges, the resulting distribution is uniform on the set of $d$-regular graphs, and occupies a constant fraction of the total probability of $\pnd$.  Thus any property which holds with high probability for $\pnd$ holds with high probability for $G_{n,d}$ as well.  

If $P, P'$ are two perfect matchings in $\pnd$, we write $P \sim P'$ if they differ by a single swap, changing $\{(a,b),(c,d)\}$ to $\{(a,c),(b,d)\}$ or $\{(a,d),(b,c)\}$.  The following martingale inequality~\cite[Theorem 2.19]{wormald-models} shows that a random variable which is Lipschitz with respect to these swaps is concentrated.
\begin{lemma}
\label{lem:wormald-concentration}
Let $c$ be a constant, and let $X$ be a random variable defined on $\pnd$ such that $|X(P)-X(P')| \le c$ whenever $P \sim P'$.  Then
\[
\Pr[|X - \expected[X]| > t] \le 2 \e^{-\frac{t^2}{dnc}} \, . 
\]
\end{lemma}
Now fix $\theta$, and define $X$ as the minimum number of edge constraints $P_{ij} = 0$ in the SDP~\eqref{eq:lovasz-primal} violated by an otherwise feasible solution with $\kappa = \theta$.  This meets the Lipschitz condition with $c=2$.  By assumption $X=0$ with positive probability.  Lemma~\ref{lem:wormald-concentration} then implies that (say) $\expected[X] \le (1/2) \sqrt{n \log n}$, in which case $X < \sqrt{n \log n}$ with high probability.  

Let $S$ denote the set of endpoints of the violated edges.  Then there is an orthogonal representation $\{u_i\}$ of the subgraph induced by $V \setminus S$ and a unit vector $\zee$ such that $\innerprod{u_i}{\zee} = 1/\sqrt{\theta}$ and $\innerprod{u_i}{u_j} = 0$ if $(i,j) \in E$ and $i,j \notin S$.  Our goal is to ``fix'' $\{u_i\}$ on the violated edges, and if necessary on some additional vertices, to give an orthogonal representation $\{v_i\}$ for all of $G$.  

As in~\cite{achlioptas-moore-reg,luczak}, we inductively build a set of vertices $S = U_0, U_1, \ldots, U_T = U$ as follows.  Given $U_t$, let $U_{t+1} = U_t \cup \{i,j\}$ where $i,j \notin U_t$, $(i,j) \in E$, and $i$ and $j$ each have at least one neighbor in $U_t$.  We define $T$ as the step at which there is no such pair $i,j$ and this process ends.  Let $I$ denote $U$'s neighborhood, i.e., the set of vertices outside $U$ which have a neighbor in $U$.  Then $I$ is an independent set, since otherwise the process would have continued.  We make the following claim:

\begin{lemma}
With high probability, the subgraph induced by $U$ is 3-colorable.
\end{lemma}
\begin{proof}
For all $0 \le t \le T$ we have $|U_t| = 2t+|S|$.  Moreover, the subgraph induced by $U_t$ has at least $3t + |S|/2 = (3/2) |U_t| - |S|$ edges and thus average degree at least $3 - 2 |S|/|U_t|$.  On the other hand, a crude union bound shows that for any $d$ and any $\beta > 2$, there is an $\alpha > 0$ such that, with high probability, all induced subgraphs of $G$ containing $\alpha n$ or fewer vertices have average degree less than $\beta$.  Since $|S|=o(n)$ with high probability, this implies that $|U_t| \le (2+o(1)) |S|$ for all $t$, and in particular that $|U| = o(n)$.  

The same union bound then implies that with high probability the subgraph induced by $|U|$, and all its subgraphs, have average degree less than $3$.  But this means that this subgraph has no 3-core: that is, it has at least one vertex of degree less than 3, and so will the subgraph we get by deleting this vertex, and so on.  Working backwards, we can 3-color the entire subgraph by starting with the empty set and adding these vertices back in, since at least one of the three colors will always be available to them.
\end{proof}

To define our orthogonal representation, let $w$ be a unit vector such that $\innerprod{\zee}{w} = \innerprod{u_i}{w} = 0$ for all $i \notin S$; such a vector exists since $|S| \ge 2$.  Then define
\[
\zee' = \sqrt{\frac{\theta}{\theta+1}} \,\zee + \frac{1}{\sqrt{\theta+1}} \,w \, .
\]
Then $|\zee'|^2 = 1$, and $\innerprod{w}{\zee'} = \innerprod{u_i}{\zee'} = 1/\sqrt{\theta+1}$ for all $i \notin S$.  Moreover, there exist three mutually orthogonal unit vectors $y_1, y_2, y_3$ such that $\innerprod{y_j}{\zee'} = 1/\sqrt{\theta+1}$ and $\innerprod{y_j}{w} = 0$ for all $j \in \{1,2,3\}$.  This follows from the fact that the following matrix is PSD whenever $\theta \ge 3$, in which case it can be realized as the Gram matrix of $\{y_1,y_2,y_3,w,\zee'\}$:
\[
\begin{pmatrix}
1 & 0 & 0 & 0 & \frac{1}{\sqrt{\theta+1}} \\
0 & 1 & 0 & 0 & \frac{1}{\sqrt{\theta+1}} \\
0 & 0 & 1 & 0 & \frac{1}{\sqrt{\theta+1}} \\
0 & 0 & 0 & 1 & \frac{1}{\sqrt{\theta+1}} \\
\frac{1}{\sqrt{\theta+1}} & \frac{1}{\sqrt{\theta+1}} & \frac{1}{\sqrt{\theta+1}} & \frac{1}{\sqrt{\theta+1}} & 1
\end{pmatrix} \, . 
\]
Finally, let $\sigma(i) \in \{1,2,3\}$ be a proper 3-coloring of the subgraph induced by $U$.  Then the following is an orthogonal representation of $G$,
\[
v_i = \begin{cases} 
u_i & i \in V \setminus (U \cup I) \\
w & i \in I \\
y_{\sigma(i)} & i \in U \, ,
\end{cases}
\]
and $\innerprod{v_i}{\zee'} = 1/\sqrt{\theta+1}$ for all $i$.  This gives a feasible solution to the SDP~\eqref{eq:lovasz-primal} with $\kappa = \theta+1$, implying that $\vartheta(\overline{G}) \le \theta+1$.

\end{appendix}

\end{document}